\documentclass{article}

\usepackage[utf8]{inputenc}
\usepackage{lmodern}
\usepackage[T1]{fontenc}

\usepackage{amsmath,amssymb,amsthm}

\newtheorem{theorem}{Theorem}
\newtheorem{lemma}[theorem]{Lemma}
\newtheorem{proposition}[theorem]{Proposition}

\theoremstyle{remark}
\newtheorem*{remark}{Remark}
\newtheorem*{note}{Note}

\usepackage{tikz}
\usetikzlibrary{arrows.meta}
\usepackage[vlined,ruled,titlenumbered,linesnumbered]{algorithm2e}
\usepackage{cases}
\usepackage{verbatim}
\usepackage[margin=1.5em,font=small]{caption}

\newcommand\ltilde[1]{%
\hspace{.4pt}\raisebox{-2.4pt}{%
\Large$\tilde{\smash{\hspace{-.4pt}\raisebox{2.4pt}{\normalsize$#1$}}}$}}

\newcommand\ltildes[1]{%
\hspace{.4pt}\raisebox{-1.6pt}{%
\large$\tilde{\smash{\hspace{-.4pt}\raisebox{1.6pt}{\small$#1$}}}$}}

\renewcommand\Pr{\mathbf P}
\newcommand\Ex{\mathbf E}

\newcommand\U{{\mathsf u}}
\newcommand\D{{\mathsf d}}
\newcommand\F{{\mathsf f}}
\newcommand\Fc{{\mathsf {f_c}}}
\let\epsilon\varepsilon
\newcommand\abs[1]{\lvert#1\rvert}
\newcommand\bigO{\mathcal O}
\newcommand\dx{\mathit{dx}}
\newcommand\atleast[2]{\langle#1\rangle_{#2}}

\DeclareRobustCommand\ttau{\smash{\ltilde\tau}}
\DeclareRobustCommand\ttaus{\smash{\ltildes\tau}}

\newcommand\hquad{\hspace{.5em}}
\newcommand\squad{\hspace{1.5em}}

\newenvironment{algo}{\begin{algorithm}[h!t]\DontPrintSemicolon}{\end{algorithm}}

\def\mathllapinternal#1#2{%
\llap{$\mathsurround=0pt#1{#2}$}}
\def\mathllap{\mathpalette\mathllapinternal}

\newenvironment{mycases}[2]{%
\def\NL{\\\hphantom{#2}\expandafter\mathllap}%
\begin{subnumcases}{#1}
\hphantom{#2}\expandafter\mathllap}%
{\end{subnumcases}\ignorespacesafterend}

\SetKw{FE}{for each}
\SetKw{IF}{if}
\SetKwIF{Wp}{Dummy}{WpElse}{with probability}{do}{with probability}{otherwise
do}{}
\SetArgSty{}

\SetVlineSkip{0ex}

\DeclareMathOperator\flip{flip}
\DeclareMathOperator\extend{extend}
\DeclareMathOperator\recover{recover}
\DeclareMathOperator\wt{wt}
\DeclareMathOperator\lift{lift}
\DeclareMathOperator\unif{Unif}
\DeclareMathOperator\poisson{Poisson}
\DeclareMathOperator\bernoulli{Bernoulli}

\tikzstyle{walk}=[]
\tikzstyle{walks}=[scale=.3,baseline=0pt]
\tikzstyle{dot}=[minimum size=1mm,inner sep=0pt,outer sep=0pt,circle,fill]

\newcommand\meander{
.. controls +(1,4) .. ++(1.5,2.5)
.. controls +(.5,-1.5) .. ++(1,-.5)
.. controls +(.5,1) .. ++(1,-.5)
.. controls +(.5,-1.5) .. ++(1.5,.5)
}
\newcommand\excursion{
.. controls +(.75,2) .. ++(1.125,1)
.. controls +(.375,-1) .. ++(.75,-.5)
.. controls +(.375,.5) .. ++(1.125,-.5)
}

\newcommand\tail{
.. controls +(.75,1.25) .. ++(1.25,.25) .. controls +(.25,-.25) .. ++(.75,-.25)
}

\newcommand\zigzag{
.. controls +(.75,0) .. ++(1.5,-1) .. controls +(.75,-1) .. ++(1.5,-1)
}

\usepackage{url}
\usepackage{hyperref}

\title{Improving the Florentine algorithms: recovering algorithms for Motzkin
and Schröder paths}
\author{Axel Bacher}

\date\today

\begin{document}

\maketitle

\begin{abstract}
We present random sampling procedures for Motzkin and Schröder paths,
following previous work on Dyck paths. Our algorithms follow the anticipated
rejection method of the Florentine algorithms (Barcucci et al. 1994+), but
introduce a recovery idea to greatly reduce the probability of rejection. They
use an optimal amount of randomness and achieve a better time complexity than
the Florentine algorithms.
\end{abstract}

\section{Introduction} \label{sec:intro}

This paper discusses random sampling procedures for two classical families of 
lattice paths: Motzkin and Schröder paths, shown in Figure~\ref{fig:paths}. We
are interested in \emph{positive paths} (paths staying above the $x$-axis) and
\emph{excursions} (positive paths with final height zero). Together with Dyck
paths, Motzkin and Schröder paths are widely studied in combinatorics. Their
counting sequences are the Catalan, Motzkin and Schröder numbers; as can be
seen in their OEIS entries \cite{oeis} (A001405, A000108, A005773, A001006,
A026003, A006318 and related ones), they are in bijection with hundreds of
combinatorial objects---most notably, binary and unary-binary trees---making
interesting and useful the problem of their efficient random sampling.

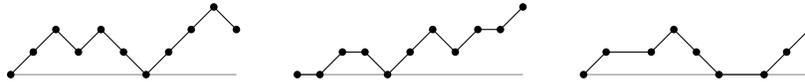
\begin{figure}[ht]\small
\newcommand\upstep{ -- ++(1,1) node [dot] {}}%
\newcommand\downstep{ -- ++(1,-1) node [dot] {}}%
\newcommand\flatstep{ -- ++(1,0) node [dot] {}}%
\newcommand\Flatstep{ -- ++(2,0) node [dot] {}}%
\newcommand\walk[1]{
    \if#1u\upstep\fi
    \if#1d\downstep\fi
    \if#1f\flatstep\fi
    \if#1F\Flatstep\fi
    \if#1;\else\expandafter\walk\fi
}
\begin{center}
\begin{tikzpicture}[walks]
\draw[help lines] (0,0) -- ++(10,0);
\draw (0,0) node [dot] {} \walk uududduuud;;
\end{tikzpicture}\hfil%
\begin{tikzpicture}[walks]
\draw[help lines] (0,0) -- ++(10,0);
\draw (0,0) node [dot] {} \walk fufduudufu;;
\end{tikzpicture}\hfil%
\begin{tikzpicture}[walks]
\draw[help lines] (0,0) -- ++(10,0);
\draw (0,0) node [dot] {} \walk uFuddFuu;;
\end{tikzpicture}\hfil%
\end{center}
\caption{Dyck, Motzkin and Schröder positive paths of length~$10$.}
\label{fig:paths}
\end{figure}

The efficiency of an algorithm is measured, of course, by its time complexity
(since all the algorithms discussed here only need negligible storage in
addition to the output, space is not an issue). Also of interest for
randomized algorithms is entropy complexity, which is a measure of the
randomness consumed by the algorithm. Our model of entropy complexity closely
follows that of \cite{knuth}, which takes its roots in Shannon's information
theory and where the unit of complexity is the random bit. This framework aims
at capturing the cost of random primitives in a realistic way and avoids
unreasonable assumptions, like having access to random real numbers in
constant time.

\bigskip

Many algorithms exist for sampling lattice paths or plane trees: Boltzmann
samplers \cite{boltzmann}, Devroye's algorithm based on the cycle
lemma~\cite{devroye}, Rémy's algorithm \cite{remy}, etc.  However, none of
these algorithms have a linear complexity for exact-size sampling: Boltzmann
samplers provide approximate size, needing costly rejection to get exact size,
while both others use an entropy of $n\log n$.

The Florentine algorithms \cite{barcucci,barcucci2,penaud} are, for their
part, linear. They use an extremely simple method called \emph{anticipated
rejection}. To sample a positive path of length~$n$, the path is drawn step by
step at random, until either the length~$n$ is reached---at which point the
path is output---or the path goes below the $x$-axis---at which point the path
is deleted and the procedure started over. This is surprisingly efficient:
sampling a positive path of length~$n$ requires, on average, to draw $2n$
random steps ($n$ for the successful run, $n$ for all the failed runs). This
is due to the fact that rejection occurs, on average, on comparatively small
paths. A detailed analysis is found in \cite{louchard,rejection}. Anticipated
rejection was also used in random sampling algorithms for classes of trees,
with similar complexities \cite{motzkin}.

In the case of Dyck paths, an improved algorithm is given in \cite{mdyck}
(actually, it works in the slightly more general case of $m$-Dyck paths). The
idea of this algorithm, used before for binary trees in \cite{motzkin}, is a
\emph{recovery} method: it follows the Florentine algorithm, but if the path
goes below the $x$-axis, a ``recovering'' procedure is used to turn it into a
random positive path, from which the algorithm is resumed. By avoiding
rejection altogether, this algorithm only consumes asymptotically $n$ random
bits---which is optimal---and reads and writes~$5n/4$ steps, better than the
Florentine algorithm.

\bigskip

In this paper, we extend this recovery idea to Motzkin and Schröder paths. We
retain a small probability of rejection, but this does not affect the
complexity: the algorithms are still optimal in terms of entropy and have
better time complexity than their Florentine counterparts.

The paper is organized as follows: Section~\ref{sec:prelim} states general
remarks useful in all models; in Sections \ref{sec:motzkin} and
\ref{sec:schroeder}, we give the algorithms for Motzkin and Schröder paths,
respectively; finally, the complexity analysis is done in
Section~\ref{sec:analysis}.

\section{Preliminaries} \label{sec:prelim}

We start by giving basic definitions and notations. We denote by $\U$, $\F$
and $\D$ the up, flat and down steps, with height~$1$, $0$ and~$-1$
respectively. A \emph{path} is a word on $\{\U,\F,\D\}$; the \emph{height} of
a path~$\omega$, denoted by~$h(\omega)$, is the sum of the heights of its
steps. A path is \emph{positive} if all its prefixes have height~$\ge0$; an
\emph{excursion} is a positive path with height~$0$; a path is
\emph{Łukasiewicz} if all its prefixes are positive except the path itself,
which has height~$<0$. We denote by $\epsilon$ the empty path.

\subsection{Unfolding and folding}

In all three models, an essential ingredient of our algorithms is a classical
bijection (see, e.g., \cite[Chapter~9]{lothaire}), which we call \emph{unfolding}. This
bijection is illustrated in Figure~\ref{fig:unfold}. Consider a Łukasiewicz
path factorized as $\sigma\tau$, with $\tau\ne\epsilon$. If $h(\sigma) = k$,
the path~$\tau$ is of the form:
\[\tau = \tau_k\D\dotsm\tau_0\D\text,\]
where the $\tau_i$'s are excursions (Figure~\ref{fig:unfold}, left). Define:
\begin{equation} \label{tilde}
\ttau = \U\tau_k\dotsm\U\tau_0\text.
\end{equation}

\begin{figure}[ht]\small
\begin{center}
\begin{tikzpicture}[walks]
\draw [help lines] (0,0) -- ++(16,0);
\draw [walk] (0,0) node [dot] {} \meander node [dot] {}
\excursion node [dot] {} -- ++(1,-1) node (a) [dot] {};
\draw [densely dotted] (a) -- ++(3,-1) node (b) [dot] {};
\draw[walk] (b) \excursion node [dot] {} -- ++(1,-1) node [dot] {};
\node at (2.5,3.5) {$\sigma$};
\node at (6.5,4.25) {$\tau_k$};
\node at (13.5,2.25) {$\tau_0$};

\draw[->] (17,.5) -- ++(2,0);

\begin{scope}[xshift=20cm]
\draw[help lines] (0,0) -- ++(16,0);
\draw[walk] (0,0) node [dot] {} \meander node [dot] {}
-- ++(1,1) node [dot] {} \excursion node (a) [dot] {};
\draw [densely dotted] (a) -- ++(3,1) node (b) [dot] {};
\draw (b) -- ++(1,1) node [dot] {} \excursion node [dot] {};
\node at (2.5,3.5) {$\sigma$};
\node at (7.5,5.25) {$\tau_k$};
\node at (14.5,7.25) {$\tau_0$};
\end{scope}
\end{tikzpicture}
\end{center}
\caption{The unfolding operation, turning the Łukasiewicz path $\sigma\tau$
(left) into the positive path of odd height $\sigma\ttaus$ (right).}
\label{fig:unfold}
\end{figure}
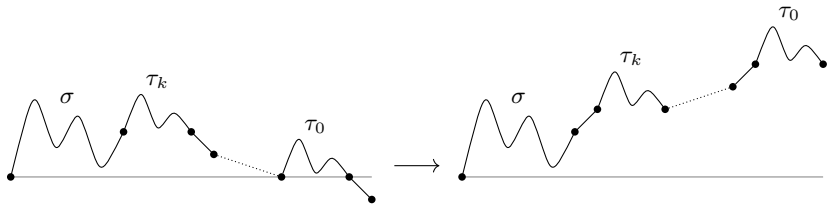

\begin{proposition} \label{prop:unfold}
Every positive path of odd height can be written in a unique way as
$\sigma\ttau$, where $\sigma\tau$ is a Łukasiewicz path.
\end{proposition}

\begin{proof}
Let $\omega$ be a positive path of height~$2k+1$. According to the definition
\eqref{tilde}, it can be written $\sigma\ttau$ only if $\ttau$ is the suffix
going up to the last visit at height~$h(\sigma)$. Moreover, the path
$\sigma\tau$ is Łukasiewicz if and only if $h(\sigma) = k$. This shows the
proposition.
\end{proof}

In the following, we call \emph{mid-height factorization} of~$\omega$ the
factorization~$\sigma\ttau$.

\begin{remark}
For the purpose of this paper, it is also acceptable, and perhaps simpler, to
define $\ttau$ as the \emph{mirror} of $\tau$ (read $\tau$ backwards and
change every $\U$ to a $\D$ and vice versa). We prefer the definition
\eqref{tilde} because it seems more robust theoretically (the mirror does not
work in the case of $m$-Dyck paths discussed in \cite{mdyck}) and because it
only involves reading $\tau$ once in the forward direction, which is better in
practice.
\end{remark}

\subsection{Structure of a recovering algorithm}

Like in the Florentine algorithms, recovering algorithms build a path by
adding random steps drawn according to some basic distribution. They also use
a function, which we denote by $\recover$, operating from Łukasiewicz paths to
positive paths. The general structure is as follows.

\begin{algo}
\caption{Recovering algorithm for a random positive path of length~$n$}
\label{algo:struct}
$\omega\leftarrow\epsilon$\;
\While{$\abs\omega < n$}{
    add a random step to~$\omega$\;
    \lIf{$h(\omega) < 0$}{$\omega\leftarrow\recover(\omega)$}
}
\Return{$\omega$}
\end{algo}

If $n\ge0$, consider the random path when it reaches a length at least~$n$ for
the first time. Let $\mathcal P_n$ be the distribution of that path
conditioned to be positive and $\mathcal L_n$ be the distribution of that path
conditioned to be Łukasiewicz.

\begin{theorem} \label{thm:recover}
Assume that the $\recover$ function, when its input is distributed
like~$\mathcal L_n$, outputs a path distributed like $\mathcal P_n$. Then
Algorithm~\ref{algo:struct} outputs a path distributed like $\mathcal P_n$.
\end{theorem}

\begin{proof}
By induction, assume that the path~$\omega$ is
distributed like $\mathcal P_{n-1}$ when it first reaches a length~$\ge n-1$.
When it first reaches a length~$\ge n$, it is either positive, in which case
it is distributed like~$\mathcal P_n$, or Łukasiewicz, in which case it is
distributed like~$\mathcal L_n$. After recovering, it is therefore distributed
like~$\mathcal P_n$.
\end{proof}

In the case of Dyck paths \cite{mdyck}, the recover function works by taking a
random factorization $\omega = \sigma\tau$ and unfolding; the result is a
uniformly distributed positive path by Proposition~\ref{prop:unfold}.

In the cases of Motzkin and Schröder paths presented in this paper, the
recover function does not work so perfectly: we retain some measure of
anticipated rejection. To represent this, we define it as a \emph{partial}
function, meaning that it may be undefined with some probability. By
convention, whenever an algorithm computes an undefined result, it immediately
rejects the sample and terminates. The algorithm is then restarted until it
produces an output. For the recovering algorithm to work, the output of the
recover function only needs to follow the distribution~$\mathcal P_n$ when it
is defined.

If $\omega$ is a path, denote by $\atleast\omega k$ the path $\omega$ if its
height is at least~$k$ and undefined otherwise.

\section{Motzkin paths} \label{sec:motzkin}

In the Motzkin case, we build the path~$\omega$ by adding $\U$, $\D$ and $\F$
steps with probability~$1/3$, $1/3$ and $1/3$. The distributions $\mathcal
P_n$ and $\mathcal L_n$ are the uniform distributions on positive and
Łukasiewicz paths of length~$n$, respectively.

\subsection{The recover operation}

The difficulty in constructing a recovering procedure for Motzkin paths is the
fact that, for any given~$n$, there are Motzkin positive paths of length~$n$
of both odd and even height. Since unfolding only produces paths of odd
height, we need a way to turn a path of odd height into one of even height.
This is the purpose of the following involution, defined for paths which are
not all $\D$ steps:
\begin{equation}
\flip\colon\sigma\F\D^k\longleftrightarrow\sigma\U\D^k\text.
\end{equation}
This operation is illustrated in Figure~\ref{fig:flip}.

\begin{figure}[ht]\small
\begin{center}
\begin{tikzpicture}[walks]
\draw[walk] (0,0) \tail node [dot] {} -- ++(1,0) node [dot] {} -- ++(1,-1)
coordinate (a);
\draw[densely dotted] (a) -- ++(1,-1) coordinate (b);
\draw[walk] (b) -- ++(1,-1) node [dot] {};
\draw[<->] (7,.5) -- ++(2,0);
\begin{scope}[xshift=10cm]
\draw[walk] (0,0) \tail node [dot] {} -- ++(1,1) node [dot] {} -- ++(1,-1)
coordinate (a);
\draw[densely dotted] (a) -- ++(1,-1) coordinate (b);
\draw[walk] (b) -- ++(1,-1) node [dot] {};
\end{scope}
\end{tikzpicture}
\end{center}
\caption{The $\flip$ operation, which changes the parity of the final height.}
\label{fig:flip}
\end{figure}
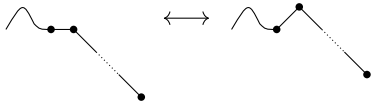

Let $q_n = 1/(2n+1)$. We define the $\recover$ operation, which takes a
Łukasiewicz path of length~$n$ and outputs a random positive path as follows:
\begin{mycases}{\recover\colon}{\sigma\tau}
{\sigma\tau}\mapsto\sigma\ttau&with proba.\hquad$q_n$, \label{ra}\NL
{\sigma\tau}\mapsto\flip(\sigma\ttau)&with proba.\hquad$q_n$, \label{rb}\NL
{\omega}\mapsto\atleast{\flip(\omega)}0&with proba.\hquad $q_n$. \label{rc}
\end{mycases}
Since there are $n$ possible factorizations $\omega = \sigma\tau$ with
$\tau\neq\epsilon$, the cases \eqref{ra} and \eqref{rb} are taken with
probability~$n/(2n+1)$ each for any given~$\omega$. Thus,
if~$\atleast{\flip(\omega)}0$
is undefined (either $\omega = \D$ or $\omega$ ends with $\U\D^k$), then
rejection occurs with probability $q_n$.

\begin{lemma} \label{lem:motzkin}
Assume that $\omega$ is equal to every Łukasiewicz path of length~$n$ with
probability~$p$. Then, for every positive path $\omega'$ of length~$n$, we
have:
\[\Pr\bigl[\recover(\omega) = \omega'\bigr] = pq_n\text.\]
\end{lemma}

\begin{proof}
We distinguish three cases: if $\omega'$ has odd height, it can only be built
with \eqref{ra}; if $\omega'$ has even height and $\flip(\omega')$ is
positive, it can only be built by \eqref{rb}; if $\omega'$ has height zero and
$\flip(\omega)$ is Łukasiewicz, it can only be built by \eqref{rc}. In all
three cases, by Proposition~\ref{prop:unfold}, there is only one way to build
$\omega'$; thus, it is output with probability $pq_n$.
\end{proof}

\subsection{Main algorithms}

We are now ready to write the algorithms sampling random Motzkin positive
paths and excursions. We only give here the proofs of their correction;
complexity is discussed in Section~\ref{sec:analysis}.

\begin{algo}
\caption{Random Motzkin positive path of length~$n$} \label{algo:motzkin}
$\omega\leftarrow\epsilon$\;
\For{$i=1,\dotsc,n$}{
    $\omega\leftarrow\omega\U$, $\omega\F$ or $\omega\D$ with probabilities
$1/3$, $1/3$ and $1/3$\;
    \lIf{$h(\omega) = -1$}{$\omega\leftarrow\recover(\omega)$}
}
\Return{$\omega$}
\end{algo}

\begin{algo}
\caption{Random Motzkin excursion of length~$n$} \label{algo:motzkinexc}
$\omega\leftarrow$ random Motzkin positive path of length~$n+1$\;
\lIf{$h(\omega)$ is even}{$\omega\leftarrow\atleast{\flip(\omega)}1$}
$\sigma\ttau\leftarrow$ mid-height factorization of $\omega$\;
\Return{$\sigma\tau$ minus the last $\D$ step}
\end{algo}

The uniformity of the output of Algorithm~\ref{algo:motzkin} is an immediate
consequence of Theorem~\ref{thm:recover} and Lemma~\ref{lem:motzkin}.
To show that the excursion output by Algorithm~\ref{algo:motzkinexc} is
uniform, let $p$ be the probability of any positive path of length~$n+1$ to be
drawn at line~1. After line~2, the path~$\omega$ is equal to every positive
path of odd height with probability~$2p$ and, after line~3, to every
Łukasiewicz path with probability~$2(n+1)p$.

\subsection{Colored Motzkin paths}

In this section, we consider Motzkin paths where the flat step carries a given
positive real weight (this may be the case if there are several kinds of flat
steps, hence the name colored Motzkin paths). We call \emph{weight} of a
path~$\omega$ and denote by $\wt(\omega)$ the product of the weights of its
steps. Florentine algorithms for these paths are discussed in
\cite{barcucci2}.

It is also possible to define colored Motzkin paths with a weight for the up
step, but this is not as interesting: if that weight is $>1$ (positive drift),
paths naturally go away from the $x$-axis and the Florentine algorithm is
already asymptotically optimal; if it is $<1$ (negative drift), paths
naturally go below the $x$-axis and the Florentine algorithm is exponential.
If we are interested in excursions, we do not lose any generality by assuming
that the weight of~$\U$ is~$1$. In any case, we need this condition so that
the unfold function does not change the weight of the path.

We further impose that the weight of the $\F$ step is greater than~$1$. In
this case, we may assume that there are four kinds of steps: $\U$, $\F$ and
$\D$, with weight~$1$ each, and a fourth kind, $\Fc$, with a weight~$c > 0$.
If $\omega$ is a colored Motzkin path, define the \emph{flippable step}
(FS) of~$\omega$ to be the last $\U$ or $\F$ step, if it exists. Let
$\flip(\omega)$ be the path obtained by changing the FS from $\U$ to $\F$ or
vice versa.

Let $q_n = 1/[2n+\max(1,c)]$. Define the new recover function as follows:
\begin{mycases}{\recover\colon}{\omega\D}
{\sigma\tau}\mapsto\sigma\ttau&
\text{with proba.\hquad$q_n$,} \label{rca}\NL
{\sigma\tau}\mapsto\flip(\sigma\ttau)&
\text{with proba.\hquad$q_n$,} \label{rcb}\NL
{\omega}\mapsto\flip(\omega)&
\text{with proba.\hquad$q_n\hphantom c$\squad if FS $ = \F$,}
\label{rcc}\NL
{\omega\D}\mapsto\omega\Fc&
\text{with proba.\hquad$q_nc$\squad otherwise.}
\label{rcd}
\end{mycases}
By construction, the probabilities sum to at most~$1$. Rejection occurs for
some paths when $c\ne1$. When $c = 1$, there is no rejection; this is the case
of bicolored Motzkin paths, which are famously counted by the Catalan numbers.
The last two cases are illustrated in Figure~\ref{fig:flips}.

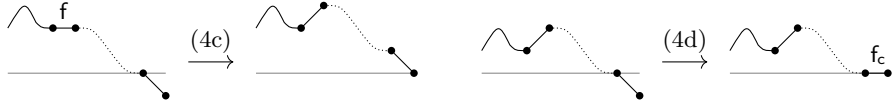
\begin{figure}[ht]\small
\begin{center}
\begin{tikzpicture}[walks]
\draw [help lines] (0,0) -- ++(7,0);
\draw (0,2) \tail node [dot] {} -- node [above] {$\F$} ++(1,0) node (a) [dot] {};
\draw[densely dotted] (a) \zigzag node (b) [dot] {};
\draw (b) -- ++(1,-1) node [dot] {};
\draw[->] (8,.5) -- node [above] {\eqref{rcc}} ++(2,0);

\begin{scope}[xshift=11cm]
\draw [help lines] (0,0) -- ++(7,0);
\draw (0,2) \tail node [dot] {} -- ++(1,1) node (a) [dot] {};
\draw[densely dotted] (a) \zigzag node (b) [dot] {};
\draw (b) -- ++(1,-1) node [dot] {};
\end{scope}

\begin{scope}[xshift=21cm]
\draw [help lines] (0,0) -- ++(7,0);
\draw (0,1) \tail node [dot] {} -- ++(1,1) node (a) [dot] {};
\draw[densely dotted] (a) \zigzag node (b) [dot] {};
\draw (b) -- ++(1,-1) node [dot] {};
\draw[->] (8,.5) -- node [above] {\eqref{rcd}} ++(2,0);

\begin{scope}[xshift=11cm]
\draw [help lines] (0,0) -- ++(7,0);
\draw (0,1) \tail node [dot] {} -- ++(1,1) node (a) [dot] {};
\draw[densely dotted] (a) \zigzag  node (b) [dot] {};
\draw (b) -- node [above] {$\Fc$} ++(1,0) node [dot] {};
\end{scope}
\end{scope}
\end{tikzpicture}
\end{center}
\caption{Left: the case \eqref{rcc}, producing excursions with a $\U$ FS
ending with $\D$. Right: the case \eqref{rcd}, producing excursions with a
$\U$ FS ending with~$\Fc$ (as well as the excursion with all $\Fc$ steps). The
dotted part consists of $\D$ and $\Fc$ steps.}
\label{fig:flips}
\end{figure}

\begin{lemma}
Assume that $\omega$ is equal to every Łukasiewicz path of length~$n$ with
probability $p\wt(\omega)$. Then, for every positive path~$\omega'$ of
length~$n$, we have:
\[\Pr\bigl[\recover(\omega) = \omega'\bigr] = pq_n\wt(\omega')\text.\]
\end{lemma}

\begin{proof}
A positive path can be built by the $\recover$ function in four different
ways: paths with odd height are built with \eqref{rca}; paths with even
non-zero height and excursions with an $\F$ FS are built with \eqref{rcb};
excursions with a $\U$ FS and ending with $\D$ are built with
\eqref{rcc}; excursions with a $\U$ FS and ending with $\Fc$, plus the
excursion $\Fc^n$, are built with \eqref{rcd}.

Every path can be built exactly once in this way. Moreover, the number of
steps~$\Fc$ is not changed by $\recover$, except in the case \eqref{rcd},
which adds one $\Fc$ step and happens with a probability multiplied by~$c$.
\end{proof}

To sample a positive path or excursion, we use Algorithm~\ref{algo:motzkin} or
\ref{algo:motzkinexc} with the new $\recover$ function. The proofs are
identical.

\section{Schröder paths} \label{sec:schroeder}

In a Schröder path, flat steps have length~$2$ instead of~$1$. To avoid
ambiguities, we denote by $\abs\omega$ the number of steps of~$\omega$, by
$\abs\omega_\F$ its number of flat steps, and by $\ell(\omega) = \abs\omega +
\abs\omega_\F$ its length. Following \cite{penaud}, we build Schröder paths by
adding $\U$, $\F$ and $\D$ steps with probabilities $r$, $r^2$ and $r$, where
$r = \sqrt2-1$, satisfying $2r + r^2 = 1$, is the radius of convergence of the
generating function of Schröder paths.

Like a Dyck path, the length and height of a Schröder path have the same
parity.
When $n$ is odd, $\mathcal L_n$ is the uniform distribution on Łukasiewicz
paths of length~$n$. However, a positive path may either reach length~$n$
exactly or overstep it with a flat step. Therefore, a path distributed
like~$\mathcal P_n$ is equal to every positive path of length~$n$ with
probability proportional to~$1$ and to every positive path of length~$n+1$
ending with~$\F$ with probability proportional to~$r$.

\subsection{The recover operation}

The recover operation should take a Łukasiewicz path of length~$n$ and output
a path of length either~$n$ or~$n+1$. Our first tool is the following random
function, which extends a path of length~$n$ into a path of length~$n+1$ in a
recursive manner:
\begin{mycases}{\extend\colon}{\omega\D}
{\omega}\mapsto\omega\U&with proba.\hquad$r$, \label{ea}\NL
{\omega}\mapsto\omega\D&with proba.\hquad$r$, \label{eb}\NL
{\omega\U}\mapsto\omega\F&with proba.\hquad$r^2$, \label{ec}\NL
{\omega\D}\mapsto\omega\F&with proba.\hquad$r^2$, \label{ed}\NL
{\omega\F}\mapsto\extend(\omega)\F&with proba.\hquad$r^2$. \label{ee}
\end{mycases}
Rejection occurs with probability $r^2$ if $\omega = \epsilon$ and, because of
the recursion, with probability $r^{2k+2}$ if $\omega = \F^k$. This function
is illustrated in Figure~\ref{fig:extend}.

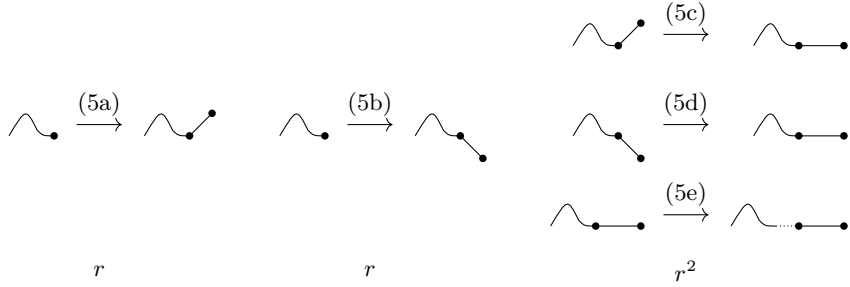
\begin{figure}[ht]\small
\begin{center}
\begin{tikzpicture}[walks]
\draw[walk] (0,0) \tail node [dot] {};
\draw[walk] (6,0) \tail node [dot] {} -- ++(1,1) node [dot] {};
\draw[->] (3,.5) -- node [above] {\eqref{ea}} ++(2,0);
\node at (4,-6) {$r$};
\begin{scope}[xshift=12cm]
\draw[walk] (0,0) \tail node [dot] {};
\draw[walk] (6,0) \tail node [dot] {} -- ++(1,-1) node [dot] {};
\draw[->] (3,.5) -- node [above] {\eqref{eb}} ++(2,0);
\node at (4,-6) {$r$};
\end{scope}
\begin{scope}[xshift=24cm]
\begin{scope}[yshift=4cm]
\draw[walk] (1,0) \tail node [dot] {} -- ++(1,1) node [dot] {};
\draw[walk] (9,0) \tail node [dot] {} -- ++(2,0) node [dot] {};
\draw[->] (5,.5) -- node [above] {\eqref{ec}} ++(2,0);
\end{scope}

\begin{scope}
\draw[walk] (1,0) \tail node [dot] {} -- ++(1,-1) node [dot] {};
\draw[walk] (9,0) \tail node [dot] {} -- ++(2,0) node [dot] {};
\draw[->] (5,.5) -- node [above] {\eqref{ed}} ++(2,0);
\node at (6,-6) {$r^2$};
\end{scope}

\begin{scope}[yshift=-4cm]
\draw[walk] (0,0) \tail node [dot] {} -- ++(2,0) node [dot] {};
\draw[walk] (8,0) \tail coordinate (a);
\draw[walk,densely dotted] (a) -- ++(1,0) node (b) [dot] {};
\draw[walk] (b) -- ++(2,0) node [dot] {};
\draw[->] (5,.5) -- node [above] {\eqref{ee}} ++(2,0);
\end{scope}
\end{scope}
\end{tikzpicture}\vspace{-.3cm}
\end{center}
\caption{The three ways of extending a Schröder path: with probability~$r$,
adding an up step~\eqref{ea}; with probability~$r$, adding a down
step~\eqref{eb}; with probability~$r^2$, turning the last up~\eqref{ec} or
down~\eqref{ed} step into a flat step, or recursively extend the part before
the last flat step~\eqref{ee}.}
\label{fig:extend}
\end{figure}

\begin{lemma} \label{lem:extend}
Assume that $\omega$ is equal to every positive path of length~$n$ with some
probability~$p$. For every positive path~$\omega'$ of length~$n+1$ and
height $>0$, we have:
\[\Pr\bigl[\extend(\omega) = \omega'\bigr] = pr\text.\]
\end{lemma}

\begin{proof}
We establish the result by induction on~$n$. If $\omega'$ ends with a $\U$
(resp.~$\D$), then it can only be built with \eqref{ea} (resp.\ \eqref{eb}),
yielding a probability~$pr$. If it ends with an $\F$, then it can be built
with \eqref{ec}, \eqref{ed} or \eqref{ee}, yielding a probability~$pr^2 + pr^2
+ pr^3 = pr$ by induction hypothesis. Note that the condition $h(\omega') > 0$
is necessary in order to build $\omega'$ using \eqref{ed}.
\end{proof}

Let $q_n = 1/(n+r)$. We define the $\recover$ function in the following way:
\begin{mycases}{\recover\colon}{\sigma\F\tau}
{\sigma\tau}\mapsto\sigma\ttau&with proba.\hquad$q_n$, \label{rsa}\NL
{\sigma\F\tau}\mapsto\atleast{\extend(\sigma\ttau)}2\,\F&with proba.\hquad$q_n$,
\label{rsb}\NL
{\omega\D}\mapsto\omega\F&with proba.\hquad$q_nr$.
\label{rsc}
\end{mycases}
To legitimize this definition, we note that there are $\abs\omega$
factorizations of~$\omega$ of type $\sigma\tau$ and $\abs\omega_\F$
factorizations of type~$\sigma\F\tau$; those two numbers sum to~$n$.
Rejection occurs in the case~\eqref{rsb} when $\extend(\sigma\ttau)$ has
height~$0$.

\begin{lemma} \label{lem:schroeder}
Let $n$ be odd and let $\omega$ be equal to every Lukasiewicz path of
length~$n$ with probability~$p$. Then, for every positive path $\omega'$ of
length~$n$, we have:
\[\Pr\bigl[\recover(\omega) = \omega'\bigr] = pq_n\]
and for every positive path~$\omega'$ of length~$n+1$ ending with~$\F$, we
have:
\[\Pr\bigl[\recover(\omega) = \omega'\bigr] = pq_nr\text.\]
\end{lemma}

\begin{proof}
If $\ell(\omega') = n$, then $\omega'$ can only be produced with \eqref{rsa},
with a probability~$pq_n$. If $\ell(\omega') = n+1$, then if $h(\omega') > 0$,
it can only be produced by \eqref{rsb}, and Lemma~\ref{lem:extend} shows that
it occurs with probability~$pq_nr$; if $h(\omega') = 0$, it can only be
produced by \eqref{rsc}, with probability $pq_nr$.
\end{proof}

\begin{comment}
\begin{remark}
The $\recover$ operation presented here is not as optimized as could be; we
are, in a sense, rejecting too often. However, optimizing it further would
only give a negligible benefit in the final complexity, so we preferred to
present this version for simplicity. A better version is given in the
Appendix.
\end{remark}
\end{comment}

\subsection{Main algorithms}

We are finally ready to present our algorithms for random Schröder positive
paths and excursions, which are a bit more involved than in the Motzkin case.
Again, analysis is postponed to Section~\ref{sec:analysis}.

\begin{algo}
\caption{Random Schröder positive path of length~$n$ or $n-1$}
\label{algo:schroeder}
$\omega\leftarrow\epsilon$\;
\While{$\ell(\omega) < n$}{
    $\omega\leftarrow\omega\U$, $\omega\F$ or $\omega\D$ with probabilities $r$, $r^2$ and
      $r$\;
    \lIf{$h(\omega) = -1$}{$\omega\leftarrow\recover(\omega)$}
}
\lIf{$\omega = \sigma\F$, $\ell(\sigma) = n-1$}{$\omega\leftarrow\sigma$}
\Return{$\omega$}
\end{algo}

As a consequence of Theorem~\ref{thm:recover} and Lemma~\ref{lem:schroeder},
the output of Algorithm~\ref{algo:schroeder} is either a path of length~$n$
with some probability~$p$ or a path of length~$n-1$ with probability~$pr$.

\begin{algo}
\caption{Random Schröder positive path of length~$n$ (odd)}
\label{algo:schroederodd}
$\omega\leftarrow$ random positive path of length~$n$ or $n-1$\;
\lIf{$\ell(\omega) = n-1$}{$\omega\leftarrow\atleast{\extend(\omega)}1$}
\Return{$\omega$}
\end{algo}

\begin{algo}
\caption{Random Schröder excursion of length~$n$ (even)}
\label{algo:schroederexc}
$\omega\leftarrow$ random positive path of length~$n$ or $n-1$\;
\uIf{$\ell(\omega) = n$}{
    $\omega\leftarrow\atleast{\extend(\omega)}1$\;
    $\sigma\ttau\leftarrow$ mid-height factorization of~$\omega$\;
    \Return{$\sigma\tau$ minus the last $\D$ step}
}
\Else{
    $\sigma\ttau\leftarrow$ mid-height factorization of~$\omega$\;
    \Return{$\sigma\F\tau$ minus the last $\D$ step}
}
\end{algo}

\begin{algo}
\caption{Random Schröder positive path of length~$n$ (even)}
\label{algo:schroedereven}
\uWp{$(n+1)/(n+1+r)$}{
    $\omega\leftarrow$ random positive path of length~$n$ or $n-1$\;
    \lIf{$\ell(\omega) = n-1$}{$\omega\leftarrow\atleast{\extend(\omega)}2$}
}
\Wp{$r/(n+1+r)$}{
    $\omega\leftarrow$ random excursion of length~$n$\;
}
\Return{$\omega$}
\end{algo}

In Algorithm~\ref{algo:schroederodd}, by Lemma~\ref{lem:extend}, every path of
length~$n$ is output with probability $p(1+r^2)$.

In Algorithm~\ref{algo:schroederexc}, if the \textbf{if} branch is taken,
extending produces every positive path of length~$n+1$ with probability~$pr$.
After folding, we get every Łukasiewicz path of length~$n+1$ with a
probability~$pr$ for each factorization of type~$\sigma\tau$ and each
factorization of type~$\sigma\F\tau$, which adds up to $pr(n+1)$ for each
path.

In Algorithm~\ref{algo:schroedereven}, finally, the branch in lines 1--3
produces every positive path with probability $p(n+1)/(n+1+r)$ and every
non-excursion positive path with probability $pr^2(n+1)/(n+1+r)$. The second
branch produces every excursion with probability $pr(n+1)r/(n+1+r)$. In total,
we get every positive path with probability $p(1+r^2)(n+1)/(n+1+r)$.

\subsection{Little Schröder paths}

Our last algorithms concern \emph{little Schröder paths}, which are Schröder
paths without flat steps at height~$0$ (entries A247623 and A001003 of the
OEIS). Non-little and little paths are closely linked, as evidenced by the
following classical bijection: by decomposing at the first flat step at
height~$0$, any non-little Schröder path can be written in a unique way as
$\sigma\F\tau$, where $\sigma$ is a little excursion. Define:
\[\lift\colon\sigma\F\tau \mapsto \sigma\U\tau\text.\]
Obviously, $\lift$ is a bijection between non-little paths of length~$n$ and
little paths of length~$n-1$ and height~$\ge1$. Moreover,
$\omega\mapsto\lift(\omega)\D$ is a bijection between non-little and little
excursions. This shows the well-known fact that the number of Schröder
excursions of any positive length is exactly twice the number of little
Schröder excursions and immediately gives Algorithm~\ref{algo:littleexc}.

\begin{algo}
\caption{Random little Schröder excursion of length~$n$ (even)}
\label{algo:littleexc}
$\omega\leftarrow$ random excursion of length~$n$\;
\lIf{$\omega$ is not little}{
    $\omega\leftarrow\lift(\omega)\D$
}
\Return{$\omega$}
\end{algo}

Sampling little positive paths is a bit more complicated. We need the
following lemma, which establishes how the $\extend$ operation behaves with
respect to little Schröder paths.

\begin{lemma} \label{lem:extendl}
Assume that $\omega$ is equal to every little positive path of length~$n$ with
probability~$p$. For every little positive path~$\omega'$ of
length~$n+1$, unless $\omega'$ has height~$1$ and ends with~$\F$, we have:
\[\Pr\bigl[\extend(\omega) = \omega'\bigr] = pr\text.\]
\end{lemma}

\begin{proof}
We proceed similarly to Lemma~\ref{lem:extend}, by induction on~$n$. If
$\omega'$ ends with $\U$ or $\D$, it has probability~$pr$. If it ends
with~$\F$, then it has height at least~$2$ (since it is a little path not
ending with~$\F$ at height~$1$). Write $\omega' = \sigma'\F$. Then $\omega$
can be $\sigma'\U$, $\sigma'\D$ or $\sigma\F$ when $\extend(\sigma) =
\sigma'$. In the latter case, $\sigma$ has height at least~$1$, which shows
that $\sigma\F$ is a little path. By induction hypothesis, we get a
probability of $pr^2 + pr^2 + pr^3 = pr$.
\end{proof}

\begin{algo}
\caption{Random little Schröder positive path of length~$n$ (even)}
\label{algo:littleeven}
$\omega\leftarrow$ random positive path of length~$n$\;
\If{$\omega$ is not little}{
    $\omega\leftarrow\extend(\lift(\omega))$\;
    \lIf{$\omega$ is not little}{reject}
}
\Return{$\omega$}
\end{algo}

\begin{algo}
\caption{Random little Schröder positive path of length~$n$ (odd)}
\label{algo:littleodd}
$\omega\leftarrow$ random little positive path of length~$n-1$\;
$\omega\leftarrow\extend(\omega)$\;
\lIf{$\omega = \sigma\F$ and $h(\omega) = 1$}{reject}
\lIf{$\omega = \sigma\D\D$ and $h(\omega) = -1$}{$\omega\leftarrow\sigma\F$}
\Return{$\omega$}
\end{algo}

To show the uniformity of the output of Algorithm~\ref{algo:littleeven}, note
that since $n$ is even, every positive path of length~$n-1$ has height $>0$.
Therefore, $\lift\omega$ is a uniformly distributed little positive path.
Moreover, no path of length~$n$ has height~$1$, so by Lemma~\ref{lem:extendl},
$\extend(\lift\omega)$ is also uniform when it is a little path.

Let $p$ be the probability of any path to be output by
Algorithm~\ref{algo:littleeven}. In Algorithm~\ref{algo:littleodd}, since $n$
is odd, the result of $\extend(\omega)$ is either a little positive path or a
little Łukasiewicz path. In the latter case, it necessarily ends with~$\D\D$
(since $\omega$ is a little path); every such path appears with
probability~$pr$. Lemma~\ref{lem:extendl} therefore shows the uniformity of
the output.

\begin{note}
For both excursions and positive paths, the lift operation can be made to work
in constant time: throughout all operations in the algorithms for Schröder
paths---adding a step, extending, unfolding---it is possible to keep track of
the first $\F$ step at height~$0$, if any, without any significant extra cost.
\end{note}

\section{Complexity analysis} \label{sec:analysis}

We analyse the algorithms presented in this paper in time and entropy
complexity, going as far as limit law analysis. For entropy complexity, we use
a slightly modified version of the model of \cite{knuth}: we assume that the
algorithms have access to a primitive giving a random step (a fair coin toss
in the Dyck case, a fair $3$-sided die in the Motzkin case, a $3$-sided die
with probabilities $r$, $r$ and $r^2$ in the Schröder case), which carries a
cost equal to the entropy of its distribution. Since the algorithms do not
make any expensive computations (outside of drawing random steps, which we
already accounted for), we choose for our measure of time complexity the
number of steps in memory read or written.

Define the \emph{time factor} of an algorithm to be its time complexity
divided by the number of steps of the output and, similarly, the \emph{entropy
factor} to be the entropy complexity divided by the entropy of the output.
In each case---positive paths, excursions, exact-size or not---the entropy of
the output is asymptotic to the entropy of the random path of length~$n$,
which is equal to the entropy of the random steps needed to generate it.

Let~$P$ be an inhomogeneous Poisson point process on $(0,1]$ with density
function $\lambda(x) = 1/(2x)$. Let $S$ be the sum, for all $x\in P$, of
independent variables distributed like~$\unif[0,x]$ (this law is the same as
in \cite{mdyck}, and more information---moments, density, tail
distribution---can be found in that paper; note that $S$ is well-defined
because, almost surely, the set $P$ is infinite but summable). Finally, let
$U$ be a random variable distributed like $\unif[0,1]$ independent from~$S$.

\begin{theorem}
Recovering algorithms for positive paths (Algorithms~\ref{algo:motzkin},
\ref{algo:schroeder}, \ref{algo:schroederodd}, \ref{algo:schroedereven},
\ref{algo:littleeven} and \ref{algo:littleodd}) have an entropy factor tending
to~$1$ and a time factor tending in distribution to $1 + S$ (expected
value~$5/4$). Algorithms for excursions (Algorithms~\ref{algo:motzkinexc},
\ref{algo:schroederexc} and \ref{algo:littleexc}) have an entropy factor
tending to~$1$ and a time factor tending in distribution to $1 + S + U$
(expected value~$7/4$).
\end{theorem}

For comparison, the Florentine algorithms for positive paths have, on average,
entropy and time factors of~$2$. Their limit law is discussed in
\cite{rejection}.

\begin{proof}
Let us start with Algorithms \ref{algo:motzkin} and \ref{algo:schroeder} (the
basic recovering algorithms). We also, for the moment, ignore the possibility
of rejection and consider only the final, successful run. There are two costs
to account for: the cost of drawing steps and writing them to memory (a fixed
time and entropy factor of~$1$) and the cost of recovery.

From Lemmas~\ref{lem:motzkin} and \ref{lem:schroeder}, the probability that
recovery occurs at length~$i$ is $q_i/(1+q_i) = 1/(2i) + \bigO(1/i^2)$
(Motzkin) or $1/i + \bigO(1/i^2)$ when $i$ is odd (Schröder). Moreover, since
the distribution of the path is the same whether we recover or not, recovery
occurs independently for each length. When recovery occurs at length~$i$, the
cost in entropy is $\bigO(\log i)$; the average total entropy cost is
therefore $\bigO(\log^2n)$, which is negligible.  The cost in time is the size
of a uniform random right factor of the path (the $\tau$ part of the
path~$\sigma\tau$), which gives a time factor of $\unif[0,i/n] +
\bigO(1/n)$. Let $R_n$ be the total cost of recovery and $R_{n,i}$ be the
contribution of length~$i$; let $x_i = i/n$, $y_i = x_{i+1}$ (Motzkin) or
$x_{i+2}$ (Schröder) and $S_i$ be the contribution of the interval $[x_i,y_i)$
to $S$. Under some coupling, we have:
\[
\Ex\left\lvert
\bernoulli\biggl(\frac{q_i}{1+q_i}\biggr)
-
\poisson\int_{x_i}^{y_i}\lambda(x)\dx
\right\rvert = \bigO\biggl(\frac1{i^2}\biggr)
\text.\]
%where the $\bigO(\cdot)$ is meant in the $L^1$ sense (bounded expected
%absolute value).
This implies that $R_n$ and $S$ can be coupled so that
$\Ex\left\lvert R_{n,i} - S_i\right\rvert = \bigO(1/(in))$. Since the interval
$(0,1/n)$ contributes $\bigO(1/n)$ on average to~$S$, summing over $i$ yields
$\Ex\left\lvert R_n - S\right\rvert = \bigO(\log n/n)$. Therefore, $R_n$ tends
in distribution to~$S$.

Let us now estimate the cost of the unsuccessful runs. If we reach the
length~$i-1$, the probability of rejection at length~$i$ is $\bigO(1/i)$ on
top of the probability to recover, or $\bigO(1/i^2)$. The probability of
reaching length at least $n$ is therefore:
% $\Omega(1)$, so we only reject $\bigO(1)$ times on average.
\[\prod_{i=1}^n1 - \bigO\biggl(\frac1{i^2}\biggr) = \Omega(1)\text.\]
Thus, we reject $\bigO(1)$ times on average.
Moreover, the average length at which rejection occurs is:
%if we do reject before
%length~$n$, the average length we reach is:
\[\sum_{i=1}^ni\cdot\bigO\biggl(\frac1{i^2}\biggr) = \bigO(\log n)\text.\]
Therefore, the total average cost of rejection is $\bigO(\log n)$, which is
negligible.

Finally, except in cases of probability $\bigO(1/n)$, all algorithms for
positive paths call the basic algorithm and then perform operations taking
constant time, so their complexity is the same. The algorithms for excursions,
again except in cases with probability~$\bigO(1/n)$, sample a positive path
and fold; the latter costs no entropy and entails accessing a uniformly
distributed right factor of the path, hence the result.
\end{proof}

\begin{note}
We can compute explicitly the probability of never having to reject. Let $p_n$
be the probability of any one positive path of length~$n$ to be output. In the
Motzkin case, by Lemma~\ref{lem:motzkin}, we have $p_n = p_{n-1}/3\,[1 +
1/(2n+1)]$, which entails:
\[p_n = 3^{-n}\prod_{i=1}^n\frac{2i+2}{2i+1} =
3^{-n}\frac{\Gamma(n+2)\Gamma(3/2)}{\Gamma(n+1/2)} \sim 3^n\frac{\sqrt{\pi
n}}2\text.\]
Let $M_n$ be the number of positive paths of length~$n$. Classically, we have
$M_n\sim3^{n+1/2}/\sqrt{\pi n}$. The probability of reaching length~$n$ is
therefore:
\[p_nM_n\to\frac{\sqrt3}{2}\text.\]
This means that we have a more than 86\% chance of succeeding in the first
try.

In the Schröder case, we have, by Lemma~\ref{lem:schroeder}, $p_n = p_{n-1}r[1
+ 1/(n+r)]$ if $n$ is odd and $p_n = p_{n-1}r$ if $n$ is even. Therefore, we
have:
\[p_n = r^n\prod_{i=1}^{\lceil n/2\rceil}\frac{2i+r}{2i-1+r}\sim r^n\sqrt{n/2}
\frac{\Gamma\bigl(\frac{1+r}2\bigr)}{\Gamma\bigl(1+\frac r2\bigr)}\text.\]
Let $S_n$ be the number of positive paths of length~$n$. Using the estimate
$S_n\sim 2^{-3/4}/(r^{n+1}\sqrt{\pi n})$, the probability of reaching at least
length~$n$ is:
\[p_nS_n + p_nrS_{n-1} \to \frac{2^{1/4}}{\sqrt\pi}
\frac{\Gamma\bigl(\frac{\sqrt2}2\bigr)}{\Gamma\bigl(\frac{1+\sqrt2}2\bigr)}
\text.\]
Thus, we have a more than 94\% chance of succeeding in the first try.
\end{note}

\bibliographystyle{abbrv}

\bibliography{biblio}{}

\end{document}